\documentclass[11pt]{article}

\usepackage[margin=1in]{geometry}
\usepackage{graphicx}
\usepackage{amsmath}
\usepackage{amsthm}
\usepackage{attrib}

\newtheorem{definition}{Definition}[section]
\newtheorem{theorem}{Theorem}[section]
\newtheorem{lemma}[theorem]{Lemma}

\newtheorem{corollary}[theorem]{Corollary}

\newtheorem{observation}[theorem]{Observation}
\newtheorem{claim}[theorem]{Claim}

\renewcommand{\vec}[1]{\mathbf{#1}}

\begin{document}

\nocite{*}

\title{The Value of Knowing Your Enemy}

\author{Christos Tzamos\footnote{Massachusetts Institute of Technology, Cambridge, Massachusetts, USA; \texttt{tzamos@mit.edu}. Supported by NSF Award CCF-1101491. Parts of this work were done while the author was an intern at Yahoo Labs.} \and Christopher A. Wilkens\footnote{Yahoo Labs, Sunnyvale, California, USA; \texttt{cwilkens@yahoo-inc.com}.}}

\clearpage\maketitle
\thispagestyle{empty}

\begin{abstract}


Many auction settings implicitly or explicitly require that bidders are treated equally ex-ante. This may be because discrimination is philosophically or legally impermissible, or because it is practically difficult to implement or impossible to enforce. We study so-called {\em anonymous} auctions to understand the revenue tradeoffs and to develop simple anonymous auctions that are approximately optimal.

We consider digital goods settings and show that the optimal anonymous, dominant strategy incentive compatible auction has an intuitive structure --- imagine that bidders are randomly permuted before the auction, then infer a posterior belief about bidder $i$'s valuation from the values of other bidders and set a posted price that maximizes revenue given this posterior.


We prove that no anonymous mechanism can guarantee an approximation better than $\Theta(n)$ to the optimal revenue in the worst case (or $\Theta(\log n)$ for regular distributions) and that even posted price mechanisms match those guarantees.
Understanding that the real power of anonymous mechanisms comes when the auctioneer can infer the bidder identities accurately, we show a tight $\Theta(k)$  approximation guarantee when each bidder can be confused with at most $k$ ``higher types''. Moreover, we introduce a simple mechanism based on $n$ target prices that is asymptotically optimal and build on this mechanism to extend our results to $m$-unit auctions and sponsored search.

\end{abstract}

\newpage
\setcounter{page}{1}



\section{Introduction}

\begin{quote}
So it is said that if you know your enemies and know yourself, you can win a hundred battles without a single loss.
\vspace{-1mm}

If you only know yourself, but not your opponent, you may win or may lose.
\vspace{-1mm}

If you know neither yourself nor your enemy, you will always endanger yourself.

\attrib{Sun Tzu, {\em The Art of War}}
\end{quote}
\vspace{-2mm}
In 1981, Myerson elegantly derived the revenue-optimal way to sell a single item~\cite{M81} --- each buyer's bid is transformed through a personalized virtual valuation function and then submitted to a standard second-price auction. Myerson's auction leverages precise prior beliefs in order to identify the bidder who generates the highest marginal expected revenue, allowing the seller to discriminate among bidders and extract more money from those with a higher willingness to pay.

For all its mathematical beauty, Myerson's optimal auction violates an inherently desirable property: {\em fairness}. One definition of fairness says that the auctioneer should not {\em a priori} discriminate among the auction's participants. It is a property that may be both desirable and necessary --- it is undeniably philosophically important in many applications; moreover, many settings lack a strong notion of identity, precluding explicit discrimination.

Sponsored search illustrates the practical importance and limitations of treating bidders equally ex-ante. A typical sponsored search auction run by Google, Bing, or Yahoo matches bidders to ad slots on a page of search results --- higher slots get more clicks, so higher bidders get higher slots. Suppose that the search engine identifies a group of queries where the market is thin, so the top bid is much higher than the second one. The search engine would like to enforce a premium price for the top slot; however, this effectively requires discriminating against the highest bidder.\footnote{In some sense, the search engine would like to set a reserve price for the top slot. However, this must be carefully defined when no bidder meets the reserve price or when more than one bidder meets it; the decreasing price mechanism we discuss later in this paper may be considered a natural interpretation of setting different reserve prices for different slots in a sponsored search auction.} Unfortunately, ex-ante discrimination may not be possible. Advertisers who are large will desire and demand ``fair'' treatment; due to their size, they will have the negotiating power to get it. Advertisers who are small lack the clout to demand equality; however, they are plentiful and could copy their accounts, blending into the masses to avoid explicit discrimination. As a result, search platforms like Google, Bing, and Yahoo may be prohibited from such discrimination out of necessity.

In this paper, we study the value of discriminating among your opponents in advance. Myerson's optimal auction critically requires that the seller know the identities of bidders ex-ante, so that he can price discriminate among them --- our goal is to quantify tradeoff inherent in requiring {\em ex-ante fairness} in dominant strategy incentive compatible auctions.

\paragraph{Anonymous Mechanism Design.} An {\em anonymous auction}  treats all bidders equally ex-ante. While the auctioneer may know information about the kinds of bidders who will participate --- even knowing precise prior beliefs about bidders' values --- this information cannot ex-ante be used to discriminate among them. Alternatively, one may say that the auctioneer knows precise priors but does not know which prior belongs to which bidder. Technically, an auction is anonymous if and only if it is symmetric in the sense that permuting bids will analogously permute allocations and prices.

To see the potential power of anonymous mechanisms, consider the following example: two bidders have values $v_1=\$2$ and $v_2=\$1$ for a digital good, and the auctioneer knows these values precisely. The optimal mechanism gives an item to each bidder, charges the first bidder $p_1=\$2$ and the second bidder $p_2=\$1$ for a total revenue of $\$3$. What can anonymous mechanisms do? A simple posted price will make revenue at most $\$2$, but the following mechanism will extract the optimal revenue: if one bidder bids $\$2$ and another bids $\$1$, give items to both and charge each bidder her value for a total revenue of $\$3$; if both bidders bid $\$2$, give items to both and charge them both $\$1$ for a revenue of $\$2$; otherwise, do not give anyone anything. It is easy to check that this mechanism is both symmetric and incentive compatible.

We begin by characterizing the optimal anonymous auction that is dominant strategy incentive compatible and ex-post individually rational. We show that it has a simple intuition in the digital goods setting: (1) imagine that bidders are relabeled uniformly at random before participating in the auction, then (2) use $\vec v_{-i}$ to infer a posterior belief about $v_i$ and (3) choose a posted price for bidder $i$ that maximizes revenue given this posterior. This intuition generalizes beyond the digital goods setting when the inferred posterior is regular. Some simple cases bear mention here: if the auctioneer's prior is the same for all bidders (an IID setting) or if it is impossible to confuse bidders, the optimal anonymous auction will correctly deduce everyone's identity and coincide with the unconstrained optimal auction.

With a basic understanding of anonymous auctions in hand, we study the performance of anonymous digital-goods auctions; our results are not immediately encouraging. We begin with a single-price mechanism --- a simple and naturally anonymous auction --- and show that it offers only a $\Theta(n)$ approximation in general and a $\Theta(\log n)$ approximation when priors are regular. Moreover, we show that the above results are tight even for the class of all anonymous mechanisms: prior beliefs exist so that no anonymous auction can guarantee  revenue approximation better than $\Theta(n)$ to the revenue of Myerson's optimal auction while if bidders' values are known to be drawn from uniform distributions, we can prove a lower-bound of $\Omega(\log n)$. Together, these suggest that general anonymous mechanisms cannot achieve better asymptotic guarantees than pricing in general settings and can be very far from optimal.

Having shown that anonymity can hurt revenue substantially in the worst case, we ask whether there are particular conditions under which anonymous auctions perform well. Our characterization of the optimal mechanism gives us hope: if all bidders are almost identical or almost perfectly distinguishable, then the optimal anonymous mechanism should be close to the unconstrained optimal one. In order to formalize this observation,
we consider $k$-ambiguous distributions where each bidder can be confused with at most $k$ bidders with ``higher ranked distributions'' and show that
anonymous mechanisms can guarantee a $\Theta(k)$ approximation to the optimal revenue.

Moreover, we introduce the {\em decreasing price mechanism}, a simple mechanism that naturally generalizes single price mechanisms and matches the asymptotic guarantees of the best anonymous auction. Intuitively, the mechanism is succinctly defined by a set of $n$ prices $p_1\geq\dots\geq p_n$, where $p_i$ is the price that the $i$-th-highest bidder should pay. The decreasing price mechanism implements this idea with the minimal modifications required to maintain incentive compatibility. Notably, this auction has linear description complexity, whereas the description complexity of the true optimal anonymous mechanism may be exponential or even unbounded for continuous distributions since it might offer a wide range of different prices to a bidder depending on what others bid.


Finally, we show how our decreasing price mechanism can be extended to anonymous mechanisms for $m$-unit auctions and sponsored search with the same $\Theta(k)$ guarantee for $k$-ambiguous distributions. As motivated above, a sponsored search platform may wish to charge a premium for certain slots based on the demand profile of a market. Without the ability to discriminate among bidders, the platform may be constrained to run an anonymous auction.\footnote{Many factors, such as click-through-rates (CTRs) and relevance scores, will break symmetry in a sponsored search auction. As discussed in Ashlagi~\cite{Ash08}, these can be handled in a variety of ways, e.g. by requiring symmetry among bidders with the same CTR or score. We follow Ashlagi and consider a simple model without such parameters to avoid these complexities.}  A slight modification to our decreasing price mechanism offers a way to do this.

\paragraph{Related Work.}

Deb and Pai \cite{DP13} also study the problem of designing a revenue maximizing mechanism under the anonymity constraint. They devise a set of allocation 
and payment functions such that in equilibrium bidders pay the Myerson virtual
values of their corresponding distributions and the seller achieves revenue that matches
the optimal revenue in the unrestricted case.  Their results are only for a single item,
their mechanisms are BIC and BIR and their solution implementation is Bayes Nash
Equilibrium. In contrast, attempting to get more robust and practical results, we require our mechanisms to be dominant strategy IC and ex-post IR and our implementation in Dominant Strategies.



Ashlagi \cite{Ash08} characterizes anonymous truth-revealing position auctions. He shows that under two different notions of anonymity, namely anonymity of the allocation rule and utility symmetry, every truth-revealing position auction is a VCG position auction. His work applies to deterministic auctions and doesn't consider optimizing revenue. 

A variety of problems in the optimal auction literature employ similar ideas to reach different ends. Hartline and Roughgarden \cite{HR09} study simple mechanisms that maximize seller revenue for selling a single item. They show that when bidder distributions are regular,  a second price auction with a single reserve --- a simple anonymous mechanism --- offers a constant fraction of the revenue that is achievable by Myerson's optimal auction~\cite{M81}. Prior-independent mechanisms (e.g.~\cite{DRY10,S03}) assume values are drawn I.I.D. to infer a distribution from $\vec v_{-i}$ to approximate the optimal revenue when the prior is not known. In contrast, anonymity will only be a significant constraint when values are non-I.I.D. and the optimal auction must discriminate among them. Optimal auctions for correlated bidders also use $\vec v_{-i}$ to infer a posterior over $v_i$ (see e.g.~\cite{CM85,RT13}). We will see that the optimal anonymous auction is closely related to the optimal general auction for a particular correlated prior.


\section{Model and Preliminaries}

\label{sect:basics}

A seller has $m$ identical items to sell to $n$ bidders. Each bidder i has a private
valuation $v_i$ for getting one item. The profile of agent valuations is
denoted by $\vec v = (v_1, \dots , v_n)$. The valuations of the agents are drawn from a product
distribution $\vec F = F_1 \times \cdots \times F_n$.

A mechanism $\mathcal{M} = (\mathcal{A},\mathcal{P})$ consists of an allocation function $\mathcal{A}$ and a pricing function $\mathcal{P}$.

\begin{definition}[Anonymous Mechanisms]
A mechanism is anonymous if permuting the arguments of $\vec v$ also permutes the resulting allocations and prices.\footnote{For deterministic mechanisms, this definition is unsatisfactory because tiebreaking rules are inherently asymmetric. Ashlagi~\cite{Ash08} discusses a refined notion of anonymous mechanisms that supports tiebreaking in deterministic mechanisms. We generally allow randomized mechanisms, so we keep this definition for simplicity.}
\end{definition}

\begin{definition}[Monotone Mechanisms]
A mechanism is monotone if for all profiles $\vec v$, we have that $\mathcal{A}_i(\vec v) \ge \mathcal{A}_j(\vec v) \Leftrightarrow v_i \ge v_j$ for all $i,j$.
\end{definition}

Every agent seeks to maximize his utility $\mathcal{U}_i(\vec v) = \mathcal{A}_i(\vec v) v_i - \mathcal{P}(\vec v)$.

Throughout the paper, we are focused on mechanisms that are Dominant Strategy Incentive Compatible (DSIC) and ex-post individually rational (ex-post IR). DSIC means means that an agent cannot improve his utility (expected valuation minus price) by bidding a different valuation even if he knows all
the valuations that other agents bid.
\begin{definition}[Dominant Strategy Incentive Compatible]
A mechanism is dominant strategy incentive compatible if for all profiles $\vec v$, we have that $\mathcal{A}_i(\vec v) v_i - \mathcal{P}_i(\vec v) \ge \mathcal{A}_i(\vec v_{-i}, v'_i) v_i - \mathcal{P}_i(\vec v_{-i}, v'_i)$ for all $\vec v, \vec v',  i$.
\end{definition}
Ex-post IR means that a bidder is always better-off participating in the mechanism:
\begin{definition}[Ex-Post IR]
A mechanism is ex-post individually rational if for all profiles $\vec v$, we have that $\mathcal{A}_i(\vec v) v_i - \mathcal{P}_i(\vec v) \ge 0$
\end{definition}


\section{Optimal Anonymous Auctions}\label{sec:opt-char}

First, we study optimal anonymous auctions and show that they have a natural structure --- informally, the mechanism uses the values of others $\vec v_{-i}$ to infer a posterior belief $h$ about bidder $i$'s value, then maximizes revenue in the standard way subject to the posterior beliefs $h$ (maximizing virtual value and charging the associated single-parameter payments~\cite{M81,AT01}). In the special case of a digital goods auction, each bidder is offered the item at the optimal posted price for her inferred distribution $h$.

First, since anonymous mechanisms generate the same outcome when bidders are permuted, we observe the following:
\begin{observation}\label{obs:opt-uniform-permutation}
The optimal anonymous mechanism remains optimal if we randomly rename bidders before running the auction.
\end{observation}
Moreover, if any mechanism's prior beliefs are symmetric, then bidders can be relabeled without affecting the mechanism's expected revenue:
\begin{observation}\label{obs:opt-symmetric-beliefs}
Suppose that prior beliefs $\vec F$ are symmetric (possibly correlated). Then there exists a symmetric mechanism that maximizes revenue.
\end{observation}
These observations immediately lead to the following claim that allows us to reduce the problem of finding the optimal symmetric auction to optimizing:
\begin{claim}\label{clm:symmetric-general-reduction}Any mechanism that is optimal among DSIC and ex-post IR mechanisms for the symmetric distribution
\[g(\vec x)=\frac {1} {n!} \sum_{\pi \in \Pi(n)} \prod_{i \in N} f_i( x_{\pi_i})\]
can be transformed into a mechanism that is optimal among symmetric, DSIC, and ex-post IR auctions for the beliefs $\vec F$ by relabeling bidders according to a uniformly random permutation.
\end{claim}
Building on this claim, our characterization theorem for digital goods follows by characterizing the optimal auction for $g$:
\begin{theorem}
\label{thm:optimal}
The optimal anonymous digital goods auction offers bidder $i$ a copy of the item at the revenue-maximizing price given $h(v_i| \vec v_{-i})$, the posterior belief about $v_i$ given $\vec v_{-i}$.
\end{theorem}
For mechanisms beyond digital goods, we can apply a theorem of Roughgarden and Talgam-Cohen~\cite{RT13} to characterize the optimal auction for $g$ as long as inferred posterior $h$ is regular --- the resulting optimal mechanism will infer $h$ and maximize virtual value with respect to $h$. For details, see Appendix~\ref{sec:opt-char-appendix}.
\begin{proof} Following Claim~\ref{clm:symmetric-general-reduction}, it is equivalent to study the optimal auction for the correlated distribution $g$. We know from Myerson and others~\cite{M81,AT01}, that a normalized mechanism $\mathcal M$ will be DSIC if and only if $\mathcal A_i$ is monotone in $v_i$ and payments are given by $\mathcal P(v)=v\mathcal A(v)-\int_0^{v}\mathcal A(z)dz$. In addition, when distributions $f_i$ are independent, a clever change-of-variables 

We can thus write the expected revenue $R_i$ from bidder $i$ as
\begin{align*}
R_i&=\int_{\Re_+^n}\mathcal{P}_i(\vec v)g(\vec v)d\vec v\\
&=\int_{\Re_+^n}\left(v_i\mathcal A_i(\vec v)-\int_0^{v_i}\mathcal A_i(\vec v_{-i},z)dz\right)g(\vec v)d\vec v\\
&=\int_{\Re_+^{n-1}}\int_{\Re_+}g(\vec v_{-i},v_i)\mathcal A_i(\vec v_{-i},v_i)\left(v_i-\frac{\int_{v_i}^{\infty}g(\vec v_{-i},z)dv_i}{g(\vec v_{-i},v_i)}\right)dv_id\vec v_{-i}
\end{align*}
If we let $h(v_i|\vec v_{-i})$ denote the density of $v_i$ that we can infer given $\vec v_{-i}$, $H$ the associated CDF, and $\phi^{h|\vec v_{-i}}(v_i)$ its Myersonian virtual value, we have
\[h(v_i|\vec v_{-i})=\frac{g(\vec v_{-i},v_i)}{\int_0^\infty g(\vec v_{-i},z)dz}\quad\quad\mbox{and}\quad\quad\phi^{h|\vec v_{-i}}(v_i)=v_i-\frac{1-H(v_i|\vec v_{-i})}{h(v_i|\vec v_{-i})}\]
and can rearrange to get
\[R_i=\int_{\Re_+^{n-1}}\left(\int_0^\infty g(\vec v_{-i},z)dz\right)\int_{\Re_+}h(v_i|\vec v_{-i})\mathcal A_i(\vec v_{-i},v_i)\phi^{h|\vec v_{-i}}(v_i)dv_id\vec v_{-i}\enspace.\]
It remains to choose $\mathcal A$, which can be done in an arbitrary (monotone) way for digital goods. The inner integral $\int h\mathcal A_i\phi dv_i$ is precisely the revenue when bidder $i$ has value distributed according to $h(v_i|\vec v_{-i})$, so Myerson~\cite{M81} tells us that the optimal allocation $\mathcal A_i(\vec v_{-i}, v_i)$ is a posted price to bidder $i$ that maximizes revenue given the distribution $h$.
\end{proof}

A few noteworthy extreme cases arise when the auctioneer can identify bidder $i$ given only the bids $\vec v_{-i}$:
\begin{corollary}
If the distributions $f_i$ are point distributions (bidders' values are known precisely to the auctioneer), have non-overlapping support, or are the same for all bidders, then the optimal anonymous mechanism coincides with Myerson's optimal mechanism.
\end{corollary}
In all three cases, the posterior distribution inferred from $\vec v_{-i}$ is precisely $f_i$, therefore the auction precisely identifies each bidder and runs the optimal auction.

These results suggest that anonymous mechanisms perform best when we can differentiate among the bidders; indeed, we will see that this is necessary. In Section~\ref{sect:worst-case}, we show that the anonymity constraint substantially limits revenue even when distributions are discrete over $n$ points and that assumptions like regularity of $f_i$ are insufficient. In Section~\ref{sect:limited-overlap}, we show that the performance degrades continuously with the auctioneer's ability to differentiate among the bidders.

\section{Worst-Case Approximations}
\label{sect:worst-case}

We compare the revenue guarantees of single price and anonymous mechanisms and find that that anonymous mechanisms can do no better in the worst case.

\subsection{Single Price Mechanisms}
We first look at how well single price mechanisms for $m$-unit auctions performs compared to the optimal. 
A single price mechanism allocates items to the $m$ highest bidders with values exceeding $p$ and charges them the maximum of $p$ and the $m+1$ highest bid.
\footnote{This is a regular VCG mechanism with a reserve price $p$. }
It is easy to see that single price mechanisms can get at least a $\frac 1 n$ fraction of the revenue by 
choosing as price the Myerson reserve price of a bidder's distribution chosen uniformly at random. 
However, such a linear approximation guarantee is unavoidable as we can also show a linear 
a lower bound of $m$ for the approximation.

\begin{theorem}[Single Price for General Distributions] \label{spgeneral}
For general distributions, a single price gives a $\Theta(m)$ approximation to the optimal revenue.
\end{theorem}
\begin{proof}
Consider the case where each bidder $i$ has a value of $\frac 1 {\epsilon^i}$ with probability $\epsilon^i$ and 0 otherwise. Then, the optimal mechanism gets a revenue of at least $m$ by posting a price to each bidder equal to his high value and selling to the $m$ largest. On the other hand, charging a single price to everyone gives at most $\max_i \frac 1 {\epsilon^i} \sum_{j \ge i} \epsilon^j \le \frac 1 {1-\epsilon}$.
\end{proof}

However, when all agent distributions are regular, we can show that single price mechanisms perform much better.

\begin{theorem}[Single Price for Regular Distributions] \label{spregular}
For regular distributions, a single price gives a $\Theta(\log m)$ approximation to the optimal revenue.
\end{theorem}
\begin{proof}
To prove the theorem we will apply Theorem 4.1 from \cite{ADMW13} which states that 
running VCG with the median of each agent's distribution as a reserve price (VCG-m) gives a $4$-approximation to the optimal revenue.
Therefore, it suffices to prove that the revenue of single price mechanisms is a $\Theta(\log m)$ to that of VCG-m. 

Let $p_i$ be the median prices for each bidder and assume that $p_1 \ge p_2 \ge \dots \ge p_n$. 

The revenue of VCG-m comes from 2 different sources: reserve prices, where a bidder is charged his reserve price, and competition between bidders, where a bidder is charged the bid of someone else.

If more than half of the revenue comes from competition between bidders, setting 
a price 0 for all bidders and running a simple VCG gives a $2$-approximation. 
This is because the revenue that comes from competition in VCG-m is at most $m$ times the $m+1$ largest bid which is equal to the revenue of VCG with no reserve prices.

Otherwise, more than half of the revenue comes from charging the reserve prices to bidders. 
In this case, the revenue is at most equal to $2 \sum_{i=1}^m p_i$.
Consider a mechanism that charges each price $p_i$ with probability $q_i = (i H_m)^{-1}$. The revenue of this mechanism is $\sum_{i=1}^m q_i p_i E[\textrm{\# bids} \ge p_i]$. However, we have that $E[\textrm{\# bids} \ge p_i] \ge i/2$ since each of the first $i$ bidders has at least $1/2$ of exceeding $p_i$. This gives a revenue of
$\sum_{i=1}^m  (i H_m)^{-1} p_i (i/2) = \frac {\sum_{i=1}^m  p_i} {2 H_m} $ which is a $4 H_m$ approximation to $2 \sum_{i=1}^m p_i$.
\end{proof}

This bound is tight even for bidders coming from point distributions. Suppose that each bidder $i$ has a value of $1/i$. 
The best single price gets revenue of 1 while the optimal mechanism gets revenue $H_m = \Theta(\log m)$.

\subsection{Symmetric Mechanisms}

For general anonymous mechanisms, we show that even if we pick the best symmetric mechanism we cannot get any better asymptotic guarantees than single price for general
distributions.

\begin{theorem}[Optimal Symmetric Mechanism for General Distributions] \label{lbgeneral}
The optimal symmetric mechanism $\mathcal{M}$ gives a $\Theta(m)$ approximation to the optimal revenue for general distributions.
\end{theorem}
\begin{proof}
We revisit the construction from the Theorem~\ref{spgeneral} but lower the probability that a bidder gets a high value even further. 
Each bidder $i$ now has a value of $\frac 1 {\epsilon^i}$ with probability $\delta \epsilon^i$ and 0 otherwise. 
The optimal asymmetric mechanism gets a revenue of $n \delta$.
The optimal symmetric mechanism must charge the same price whenever there is only one bidder with a high bid. 
Let $E$ be the event that at least two bidders value the item high. 
Given $\neg E$, the mechanism is identical to a single price mechanism. So the approximation of the optimal symmetric mechanism is upper bounded by 
$\frac { \delta \frac 1 {1-\epsilon} + Pr[E] Rev[E]} { n \delta } \le  \frac 1 {n (1-\epsilon)} + \frac {Pr[E] Rev[E]} { n \delta }$. 
The theorem follows since $\frac {Pr[E] Rev[E]} { n \delta }$ goes to 0 as $\delta \rightarrow 0$.
\end{proof}

Moreover, we can show that general symmetric mechanisms cannot beat the asymptotic guarantees that single price mechanisms achieve for regular
distributions.
In fact, we can show that this is true even for uniform distributions.

\begin{theorem}[Uniform distributions counterexample] \label{lbuniform}
For uniform distributions, the best symmetric mechanism gets at most a $\Theta( \log m )$ approximation to the optimal revenue.
\end{theorem}

\subsubsection{Proof of Theorem \ref{lbuniform}}

We consider a digital goods case where there are $N = (2^n - 1) L$ agents, where $2^i L$ agents have distributions in $U[0,2^{-i}]$ for $i \in \{0,...,n-1\}$. 
We can see that the optimal asymmetric mechanism gets a revenue of $\frac {L n} 4$ by charging each agent a price at the midpoint of his distribution. 

We will now upper bound the revenue that the optimal symmetric mechanism achieves.
To do this we consider an instance where a vector of values $\vec v$ is reported. 

Let $b_i = \# \{j | v_j > 2^{-i} \}$, i.e. the number of agents with values greater than $2^{-i}$.
We will show that if all $b_i$'s are large, the optimal symmetric mechanism charges a very low price to each agent.

\begin{lemma} \label{lowprice}
If $b_i > (\frac 2 3 2^{i}-1)L + 1$ for all $i \in \{1,\dots,n-1\}$, the optimal symmetric mechanism 
charges a price lower than $2^{-(n-1)}$ to every agent.
\end{lemma}

\begin{proof}
Since we are in a digital goods setting we can apply Theorem~\ref{thm:optimal} and 
consider the distribution that the mechanism infers for an agent k's value by looking at all bids of the other agents.
The probability density of agent's $k$ value at a point $x$ given the bids $\vec v_{-k}$ of the other agents is
$h(x|v_{-k}) = \frac {1} {n!} \sum_{\pi \in \Pi(n)} f_{\pi_k}( x ) \prod_{i \neq k} f_{\pi_i}( v_i )$, which is
proportional to the number of ways to match agents to
probability distributions for the bid vector $\vec v' = (\vec v_{-k}, x)$. 

We can compute the number of ways exactly
in terms of $b'_i = \# \{j | v'_j > 2^{-i} \}$ as $\prod_{i=0}^{n-1} ((2^{i+1}-1) L - b'_i)_{b'_{i+1} - b'_i}$ where the notation $(a)_b \equiv a (a-1) ... (a-b+1)$ denotes the falling factorial
and $b'_n$ is defined to be equal to $N$. This is because the $b'_1$ agents that have values in $[1/2,1]$ can only be in the distributions $U[0,1]$ so
there are L choices for distributions which means there are $L (L-1) ... (L-b'_1+1)$ ways to match them. For the 
$b'_2-b'_1$ agents that have values in $[1/4,1/2]$, there are $3 L$ possible distributions ($L$ that are U[0,1] and and $2 L$ that are U[0,1/2])
but $b'_1$ of them are already taken so there are exactly $(3 L - b'_1)_{b'_{2} - b'_1}$ choices over all and so on.

We now show that 
$4 h(x|v_{-k}) < h(y|v_{-k})$ for $x \in (2^{-t},2^{-(t-1)})$, $y \in (2^{-(t+1)},2^{-t})$ and $1 \le t \le n-1$. That is the probability density at the interval $(2^{-t},2^{-(t-1)})$ is at most a fourth of the probability density at the interval $(2^{-(t+1)},2^{-t})$. 
Let $b'(x)$ and $b'(y)$ be the corresponding $b'$ parameters for $x$ and $y$ respectively. It is easy to see that $b'_i(x) = b'_i(y)$ for $i \neq t$ and that $b'_t(x) = b'_t(y) + 1$. We have that:
\begin{align*} 
\frac { h(x|v_{-k}) } { h(y|v_{-k}) } 
&=
\frac {\prod_{i=0}^{n-1} ((2^{i+1}-1) L - b'_i(x))_{b'_{i+1}(x) - b'_i(x)}} {\prod_{i=0}^{n-1} ((2^{i+1}-1) L - b'_i(y))_{b'_{i+1}(y) - b'_i(y)}}\\ 
&=
\frac {\prod_{i=t-1}^{t} ((2^{i+1}-1) L - b'_i(x))_{b'_{i+1}(x) - b'_i(x)}} {\prod_{i=t-1}^{t} ((2^{i+1}-1) L - b'_i(y))_{b'_{i+1}(y) - b'_i(y)}}\\ 
&= \frac {(2^{t}-1) L - b'_{t}(y)} {(2^{t+1}-1) L - b'_{t}(y)}
&\textrm{cancelling all identical terms}
\\ 
	&<  \frac {(2^{t}-1) L - (\frac 2 3 2^{t}-1)L} 
{(2^{t+1}-1) L - (\frac 2 3 2^{t}-1)L} 
&\textrm{since } b'_{t}(y) \ge b_{t} - 1 > (\frac 2 3 2^{t}-1)L
\\
	&=  \frac {\frac 1 3 2^{t}} {\frac 4 3 2^{t}} = \frac 1 4
\end{align*}

We now show that the optimal price for the inferred distribution is less than $2^{-(n-1)}$. 
Assume that this is not the case and the optimal price is $p > 2^{-(n-1)}$. 
We will show that by charging $p/2$ we get strictly more revenue. We will prove by induction that $Pr[x>p] < Pr[x>p/2]/2$ for $p \in [2^{-(n-1)},2)$. 
This is trivial to see if $p \in [1,2)$ since $Pr[x>p] = 0$ while $Pr[x>p/2] > 0$. Assume that $Pr[x>p] < Pr[x>p/2]/2$ for $p \in [2^{-i},2^{-i+1})$.
Then for $p \in [2^{-i-1},2^{-i})$ we have that:
\begin{align*}
Pr[x>p] =& Pr[x>2^{-i}] + Pr[x \in (p,2^{-i})] \\
<& \frac {Pr[x>2^{-i-1}]} 2 + Pr[x \in (p,2^{-i})] & \textrm{by the induction hypothesis} \\
<& \frac {Pr[x>2^{-i-1}]} 2 + \frac {Pr[x \in (\frac p 2,2^{-i-1})]} 2  & \textrm{since } \frac {h(x|v_{-k})} {h( x / 2 |v_{-k})} < \frac 1 4 \textrm{ for } x \in (p,2^{-i})\\
=&  {Pr[x> p / 2]}  / 2
\end{align*}

We conclude that $Pr[x>p] < Pr[x>p/2]/2$ which implies that $p Pr[x>p] < p Pr[x>p/2]/2$, i.e. the revenue we get by charging $p$ is less than charging $p/2$
if $p > 2^{-(n-1)}$.
\end{proof}

We now show that for large enough $L$ the conditions of Lemma~\ref{lowprice} are satisfied with extremely high probability.
\begin{lemma}
Let $L=2^{5n}$ and let $E$ be the event that $b_i > (\frac 2 3 2^{i}-1)L + 1$ for all $i$. Then $Pr[E] < 1-n e^{-2^{n-2}}$. 
\end{lemma}
\begin{proof}
Consider the expectation of $b_i$.
\begin{align*}
E[ b_i ]  &= \sum_j Pr[ v_j > 2^{-i} ] = L 2^{0} (1 - 2^{-i}) + L 2^{1} (1 - 2^{-i+1})
 + ... +
L 2^{i-1} (1 - 2^{-1})  \\
&= L \left( 2^i - 1 - 2^i \sum_{j=1}^i 2^{-2j} \right) 
= L \left( 2^i - 1 - 2^{i} \frac {1 - 2^{-2i}} { 3 } \right) = L \left( \frac 2 3 2^i-1 + \frac {2^{-i}} { 3 } \right) 
\end{align*}
We have that $E[ b_i ] (1 - 2^{-2n}) >  2^{5n} \left( \frac 2 3 2^i-1 + \frac {2^{-i}} { 3 } \right) -   2^{3n} \frac 2 3 2^i > 2^{5n} \left( \frac 2 3 2^i-1 \right) + 1$. Therefore,

\begin{align*}
Pr[b_i < (\frac 2 3 2^{i}-1)L + 1] &< Pr\left[b_i < E[ b_i ] \left(1 - 2^{-2n}\right) \right] \\
&\le e^{-2^{-4n} E[ b_i ] / 2} & \textrm{applying a Chernoff bound}\\
&\le e^{-2^{n-2}} & \textrm{since } E[ b_i ] \ge L / 2 = 2^{5 n - 1}
\end{align*}

By a union bound for all $n$ possible values of
$i$ we get that $Pr[E] < 1-n e^{-2^{n-2}}$.
\end{proof}

Therefore, the revenue of the optimal symmetric mechanism is at most $N 2^{-(n-1)} = L (2^n-1) 2^{-(n-1)} \le 2 L$ when event $E$ happens and at most $L n$ otherwise.
Thus, the expected revenue is at most $L (2 + n^2 e ^{-2^{n-2}})$. Since the optimal asymmetric mechanism achieves revenue $L n /4$, the approximation ratio is
$n / 8 + o(1)$. Since the number of agents is at most $N \le 2^{6 n}$, we have that $n \ge \log N / 6$. Thus the approximation ratio in terms of $N$ 
is $\frac { \log N } {48} + o(1) = \Theta(\log N) = \Theta(\log m)$ since $m = N$ in the digital goods setting.


\section{Anonymous Auctions with Limited Ambiguity}

\label{sect:limited-overlap}

In the previous section, we showed that the best anonymous auction cannot offer better worst-case revenue guarantees than single price mechanisms, even when distributions are regular or have a monotone hazard rate. In this section, we explore a key property called limited ambiguity that separates anonymous mechanisms from single price mechanisms and demonstrates their power.

\begin{definition}
Let $[a_i,b_i]$ be the support of the distribution of agent $i$ and assume without loss of generality that $a_1 \ge a_2 \ge ... \ge a_n$.
We say that the set of distributions is \emph{$k$-ambiguous} if $b_i < a_{i-1-k}$ for all $i$, i.e. a sample from the $i$-th distribution can be confused with at most $k$ distributions ahead of it.
\end{definition}

The extreme case where $k=0$ --- i.e. bidders' values are drawn from distributions with disjoint supports --- gives our first separation between general anonymous auctions and single price mechanisms. It is easy to see that single price mechanisms cannot achieve approximation ratio bounded by a function of $k$ for $0$-ambiguous distributions. Consider the single point distribution $1/i$ for each agent $i$ --- it is easy to see that the approximation ratio of any single price is $\log n$, which cannot be bounded by a function of $k$. In contrast, we showed that the optimal anonymous auction achieves the same revenue as the optimal non-anonymous auction in Section~\ref{sec:opt-char}.

In this section, we will show that anonymous mechanisms can guarantee an approximation ratio of $O(k)$ for $k$-ambiguous distributions, and that this is tight. We focus first on the case of digital goods, where $m=n$, and then extend to $m<n$ as well as to sponsored search auctions.

To show that anonymous mechanisms can achieve an $O(k)$ approximation to the optimal revenue, we construct a simple mechanism called the Decreasing Price Mechanism (DPM) that is efficiently defined by $n$ prices. We will begin with a slight variation that is not dominant strategies incentive compatible (DSIC) to motivate the choice of mechanism.

\begin{definition}[Non-DSIC Decreasing Price Mechanism]
The {\em Non-DSIC Decreasing Price Mechanism} is defined by a set of prices $p_1 \ge p_2 \ge ... \ge p_n$ and works as follows:
incoming bids are sorted in decreasing order, then bidder $i$ is offered an item at price $p_i$.
\end{definition}

This mechanism is both simple and anonymous, but unfortunately it is not DSIC, since a bidder can lower the price she pays simply by ranking lower in the ordering of bids (indeed, she can always get an item at price $p_n$ simply by placing the lowest bid). We add two key ingredients to define our DSIC decreasing price mechanism. 

The first ingredient we add limits a bidder's ability to win the item at a lower price: the auction only sells an item at price $p_i$ if it has successfully sold items at all higher prices. Consequently, for example, bidder $i+1$ must be willing to pay $p_i$ in order for bidder $i$ have a chance to win an item at a lower price. When the auction fails to sell an item at price $p_i$ and therefore stops selling more items, we call this a ``{\bf drop}'' event.

The second ingredient we add restores incentive compatibility: if a bidder could have won an item at a lower price by ranking lower in the bid order, then we automatically charge her the lower price instead. Observe that given our first modification, bidder $i$ can win an item at a lower price $p_l$ if and only if $b_j\geq p_{j-1}$ for all $j\in\{i+1,\dots,l\}$. We call this a ``{\bf chain}'' effect since there is a chain of bidders with $b_j\geq p_{j-1}$.

These two additional ingredients are the intuition for our decreasing price mechanism:

\begin{definition}[Decreasing Price Mechanism]
The {\em Decreasing Price Mechanism (DPM)} is defined by a set of prices $p_1 \ge p_2 \ge ... \ge p_n$ and works as follows:
\begin{itemize}
\item Sort bids in decreasing order.
\item Starting with $i=1$, allocate items as long as $b_i\geq p_i$, then stop allocating items.
\item Each winner $i$ is charged $p_{\underbar j(i)}$, where $\underbar j(i)$ is the smallest $j\geq i$ such that exactly $j$ bidders are bidding above $p_j$.
\end{itemize}
\end{definition}
%
%

We note that single price mechanisms are a special case of DPM where all the prices $p_1 = ... = p_n = p$. The following lemma shows several interesting properties of DPM.

\begin{lemma}
The Decreasing Price Mechanism is anonymous, ex-post IR, DSIC, and monotone in the sense that if $b_i>b_j$, then $\mathcal{A}_i(\vec b)\ge \mathcal{A}_j(\vec b)$.
\end{lemma}

\begin{proof}
It is clear that the mechanism is anonymous because it ignores any initial labeling and relabels bidders in decreasing order of their bids. The auction is individually rational because a bidder only wins if $b_i\geq p_i$ and pays a price $p_{\underbar j(i)}\leq p_i$. The claimed monotonicity property is also easy to see as the mechanism considers bids in decreasing order and allocates items only until it reaches the first bidder with $b_i<p_i$.

To see that the mechanism is DSIC, we look at an agent $i$ and show that $i$ cannot win an item at a lower price. Note that if $i$ changes her bid to $b_i'<p_{\underbar j(i)}$, then there will be $\underbar j(i)-1$ bids $\geq p_{\underbar j(i)}$ (there were exactly $\underbar j(i)$ such bids before $i$ changed her bid) and the auction must stop by the time it reaches reaches bidder $\underbar j(i)$. Thus, the auction will not sell an item for less than $p_{\underbar j(i)}$, so $i$ will not get an item. On the other hand, keeping other bids fixed, if $i$ bids $b_i\geq p_{\underbar j(i)}$, there will be exactly $\underbar j(i)$ bidders bidding $\geq p_{\underbar j(i)}$, so $i$ cannot win at a price less than $p_{\underbar j(i)}$.

%
\end{proof}

We will now show that the decreasing price mechanism achieves an approximation ratio of $O(k)$ for $k$-ambiguous distributions. To illustrate the significant ideas in the proof we will first show the statement for $k=1$ before proving the general case.


\subsection{The case of $k=1$}

For $1$-ambiguous distributions, we prove the following theorem:
\begin{theorem}
The optimal Decreasing Price Mechanism approximates the revenue of the optimal auction within a factor of $5$ for $1$-ambiguous distributions.
\end{theorem}

\begin{proof} The proof has two parts. First, we use a distribution over DPM pricing schemes to approximate the revenue contribution of agents $3$ to $n$. This distribution will have expected revenue that is a $3$-approximation to the welfare of those agents and therefore also to the revenue they contribute in the optimal auction. Second, we use our single price results to cover the revenue from the first two agents.

First, to cover the revenue contributions of agents $3$ to $n$, DPM prices are chosen as follows (the parameters $r_i$ will be chosen later):
\[p_i=\begin{cases}a_{i-1}&\mbox{with probability }r_i\\a_i&\mbox{otherwise.}\end{cases}\]
Intuitively, choosing $p_i=a_i$ is safe because $v_i\geq a_i$, whereas $p_i=a_{i-1}$ extracts more revenue at the risk of triggering a drop event that prevents selling items to bidders $>i$. We take $r_1=0$ so $p_1=a_1$.

Let $q_i$ be the probability that $v_i \ge a_{i-1}$ and define $q_1 = 0$. We define $c_i$, the conditional likelihood of a chain effect, and $d_i$, the conditional likelihood of a drop event, as follows:
\[c_i\equiv\Pr[v_i\ge a_{i-1}\mbox{ and }p_i=a_i]=(1-r_i)q_i\]
\[d_i\equiv\Pr[v_i<p_i]=r_i(1-q_i)\]

By definition of the auction, agent $i$ pays at least $a_t$ for some $t\geq i$ if and only if (a) all bidders $j\leq i$ have $v_j\geq p_j$ so that bidder $i$ wins an item, and (b) there exists a $j \in \{i+1, \dots, t+1\}$ such that exactly $j$ bidders have bids $b_j\geq p_j$. Condition (a) is equivalent to saying that a drop event does not occur among the first $i$ bidders and happens with probability $\prod_{j=1}^i (1-d_j)$. Condition (b), assuming truthfulness and using $1$-ambiguity, happens if and only if there is some $j \in \{i+1, \dots, t+1\}$ such that either $v_{j} < a_{j-1}$ or $p_{j} = a_{j-1}$, which happens precisely when $j$ does not trigger a chain effect, so the likelihood that such a $j$ exists is $ 1 - \prod_{j=i+1}^{t+1} c_j $.

Let $x_t$ denote the expected number of bidders who pay $a_t$ and $y_t=\sum_{i=1}^tx_i$ the expected number who pay at least $a_t$. We can now write $y_t$ as
\begin{align*}
y_t &\ge \sum_{i=1}^t Pr[\textrm{Agent i pays at least $a_t$}] \ge  \sum_{i=1}^t  \left[ \left (1 - \prod_{j=i+1}^{t+1} c_j \right) \prod_{j=1}^i (1-d_j) \right]\enspace.
\end{align*}
To bound this sum, we relate the $c_i$'s and $d_i$'s with the following lemma:
\begin{lemma}\label{lem:k1-cd}
We can choose $r_i$ such that $d_i \le \rho$ and $c_i \le (1-\sqrt \rho)^2$ for any $\rho \in [0,1]$.
\end{lemma}
\begin{proof}
For any such $\rho$ choose $r_i = \min ( \frac \rho {(1-q_i)}, 1 )$. We have that $d_i = (1-q_i) r_i \le \rho$. We also have that $c_i = (1-r_i) q_i$. If $r_i=1$ then $c_i = 0 \le  (1-\sqrt \rho)^2$. Otherwise $r_i = \frac \rho {(1-q_i)}$ and $c_i = (1-\frac \rho {(1-q_i)}) q_i$ which achieves a maximum value at $(1-\sqrt \rho)^2$ for $q_i = 1-\sqrt \rho$,
\end{proof}
Applying this lemma with $\rho=1/i^2$ gives $r_i$'s such that $d_i \le 1/i^2$ and $c_i \le (1-1/i)^2$ for $i \ge 2$. This makes $\prod_{j=1}^i (1-d_i) \ge \prod_{j=2}^i (1-1/i^2) = (1+1/i)/2 \ge 1/2$. Moreover, $\prod_{j=i+1}^{t+1} c_i \le \prod_{j=i+1}^{t+1} (1-1/i)^2 = (i/(t+1))^2$. Therefore,
\[y_t \ge \frac 1 2 \sum_{i=1}^t \left[1 - \left(\frac i {t+1} \right)^2 \right] = \frac 1 2 \left( t - \frac {t (t+1) (2 t + 1)} { 6 (t+1)^2 } \right) \ge \frac 1 2 (t - t / 3) \ge t/3\]

The total expected revenue of the mechanism is $\sum_{i=1}^nx_ia_i$. Since $y_t=\sum_{i=1}^tx_i\geq t/3$ for all $t$, it must be that $\sum_{i=1}^nx_ia_i\geq\sum_{i=1}^na_i/3$. Moreover, since $a_t > b_{t+2} \ge Rev[\textrm{Agent}_{t+2}]$, it follows that $\sum_{t=1}^n a_t / 3 \ge \sum_{t=3}^n Rev[\textrm{Agent}_t] / 3$, i.e. the revenue is at least $1/3$ of the optimal revenue generated by agents $3$ to $n$.

It remains to handle the revenue contributed by the first two agents. To do so, we use the single price lemma that says that a single price $p$ is a $2$-factor approximation for $2$ distributions. If we choose prices $p_1 = ... = p_n = p$ with probability $2/5$ or the pricing scheme that is defined above with probability $3/5$, we get an expected revenue of at least:
$$\frac 2 5 \left( \frac {Rev[\textrm{Agent}_1] + Rev[\textrm{Agent}_2]} 2 \right) +
\frac 3 5 \left( \frac {\sum_{t=3}^n Rev[\textrm{Agent}_t] } 3 \right) = \frac {\sum_{t=1}^n Rev[\textrm{Agent}_t] } 5$$

Since we are randomizing over DPM pricing schemes, there exists a single pricing scheme that achieves the necessary approximation. This completes the proof and shows a $5$ approximation.\end{proof}


\subsection{The general case}

For general $k$-ambiguous distributions, the following theorem shows an $O(k)$ approximation.

\begin{theorem}
The Decreasing Price Mechanism achieves an approximation ratio of $(3 e^2 + 2) k$ for $k$-ambiguous distributions.
\end{theorem}

The proof of this theorem mimics the $1$-ambiguous case. We split agents into blocks of size $k$ such that an agent in block $t$ cannot be confused with any agents in blocks $<t-1$, then a technical lemma analogous to Lemma~\ref{lem:k1-cd} bounds the drop and chain rates between blocks to achieve an $O(k)$ approximation to the revenue from blocks $3$ to $n/k$. Finally, a single price mechanism covers the revenue from the top two blocks.

---

\begin{proof} To begin, we split agents into $N = \lceil n/k\rceil$ blocks, such that block 1 contains agents $1$ through $k$, block 2 contains agents from $k+1$ to $2 k$ and so on. Notice that as previously agents in block $i$ cannot be confused with agents in blocks $<i-1$. Let $A_i$ be the lowest value an agent in block $i$ can take, i.e. $A_i = a_{i \cdot k}$.

We will first approximate the revenue contribution of blocks $3$ to $N$. The main ideas follow the 1-ambiguous proof. For each block $i$, we randomly pick a number of items $j$ to price ``high:'' the top $j$ items in block $i$ are priced at $A_{i-1}$, and the remaining $k-j$ items are priced at $A_i$. A block ``drops'' if we over-estimate the number of bidders who are willing to pay $A_{i-1}$; if block $i$ drops, then the auction will not allocate to any bidders in blocks $>i$. Similarly, a block ``chains'' if we underestimate the number of bidders who are willing to pay $A_{i-1}$; if a block chains, then the auction will not be able to charge $A_{i-1}$ to any bidder, since there will be too many bidders willing to pay $A_{i-1}$.

Formally, we set prices for each block as follows:
\begin{enumerate}
\item Sample $j$ according to the distribution $R_{i,j}$. ($\sum_{j=0}^kR_{i,j}=1$)
\item Set the prices for the first $j$ items in the block at $A_{i-1}$ and set prices for the remaining $k-j$ items at $A_i$:
\[p_i=\begin{cases}A_{\lfloor i/k\rfloor}&\mbox{if }i-\lfloor i/k\rfloor \leq j\\A_{\lceil i/k\rceil}&\mbox{otherwise.}\end{cases}\]
\end{enumerate}
We set $R_{1, \cdot}=(1,0,0,0,...,0)$ so that all agents of block $1$ are assigned a price of $A_1$.

To define chain and drop probabilities, let $Q_{i, j}$ be the probability that exactly $j$ bidders in block $i$ have value greater or equal to $A_{i-1}$. We define $Q_{1 \cdot} = (1,0,0,...,0)$. We define the associated chain probability $C_i$ as the likelihood that the number of agents in block $i$ who are willing to pay $A_{i-1}$ strictly exceeds the number of prices in the block that were set at $A_{i-1}$:
\[C_i = \sum_{j=0}^{k-1} \sum_{j'=j+1}^{k} R_{i, j} Q_{i, j'}\enspace.\]
Similarly, we define the associated drop probability $D_i$ as the likelihood that the number of agents in block $i$ who are willing to pay $A_{i-1}$ is strictly less than the number of prices that were set at $A_{i-1}$:
\[D_i =  \sum_{j=0}^{k-1}  \sum_{j'=j+1}^{k} Q_{i, j} R_{i, j'}\enspace.\]

We claim that agents in block $i$ pay at least $A_t$ for some $t \ge i$ as long as (a) no block $\leq i$ ``drops,'' and (b) at least one block $j\in\{i+1,\dots t+1\}$ does not ``chain.'' Note that if no block $\leq i$ drops, then all bidders in blocks $\leq i$ have $v\geq p$ and will therefore get allocated. This happens with probability $\prod_{j=1}^i (1-D_j)$. If any block $j\in\{i+1,\dots,t+1\}$ does not chain, then we know that the number of bidders asked to pay $A_{j-1}$ cannot be higher than the number of bidders asked; consequently, $A_{j-1}$ will be a lower-bound on the price paid by bidders in block $i$. The likelihood that at least one such block does not chain is $1 - \prod_{j=i+1}^{t+1} C_j$.

Thus, if we define $X_t$ as the expected number of blocks whose agents pay $A_t$ and $Y_t=\sum_{i=1}^tX_i$ be the expected number of blocks where all agents pay at least $A_t$, then:
\begin{align*}
Y_t &\ge \sum_{i=1}^t Pr[\textrm{Agents in block $i$ pay at least $A_t$}]  \ge  \sum_{i=1}^t  \left[ \left (1 - \prod_{j=i+1}^{t+1} C_j \right) \prod_{j=1}^i (1-D_j) \right]
\end{align*}
We use the following lemma to relate $C_i$'s and $D_i$'s.

\begin{lemma}
We can choose $R_{i, \cdot}$ such that $D_i \le \rho$ and $C_i \le 1-\rho^{1-\frac 1 {k+1}}$ for any $\rho \in [0,1]$.
\end{lemma}
\begin{proof}
Let $\hat Q_{i, j} = \sum_{z=0}^j Q_{i, z}$ and $\hat Q_{i,-1} = 0$. We consider distributions $R_{i,j}$ that take the following form:
\[R_{i,j}=\begin{cases}\min(1, \rho / \hat Q_{i, (s-1)})&\mbox{if }j=s\\
1 - R_{i, s}&\mbox{if }j=0\\
0&\mbox{otherwise}\end{cases}\]
where $s$ is a parameter that will be chosen later. Notice that for any $s$, we have that:
$$D_i =  \sum_{j=0}^{k-1}  \sum_{z=j+1}^{k} Q_{i, j} R_{i, z} = \sum_{z=1}^{k}  \sum_{j=0}^{z-1} Q_{i, j} R_{i, z} =  \sum_{j=0}^{s-1} Q_{i, j} R_{i, s} = \hat Q_{i, (s-1)} R_{i, s} \le \rho$$
and
$$C_i = \sum_{j=0}^{k-1} \sum_{z=j+1}^{k} R_{i, j} Q_{i, z} = \sum_{j=0}^{k-1}  R_{i, j} (1-\hat Q_{i, j}) \le 1 - R_{i, s} \hat Q_{i, s}$$

We are now ready to prove that there exists choice of $s$ 
such that $C_i \le 1-\rho^{1-\frac 1 {k+1}}$. We will do this by assuming the contrary, namely that $C_i >  1-\rho^{1-\frac 1 {k+1}}$ 
for any choice of $s$, and reach a contradiction. Note that this assumption immediately implies that $R_{i, j} \hat Q_{i, j} \le 1 - C_i < \rho^{1-\frac 1 {k+1}}$ for any $j$. Before we begin, note that $\hat Q_{i,-1}=0$, $\hat Q_{i,k}=1$, and $\hat Q_{i,j}$ is monotone in $j$.

We let $j^*$ be the smallest $j$ that $\rho \le \hat Q_{i, j}$ and show inductively that for any $z \ge j^*$, $$\hat Q_{i, z} <  \rho^{\frac {k - (z - j^*)} {k+1}}\enspace.$$

{\em Base case $z=j^*$:} When $z=j^*$, we can choose $s=z$. Since $\hat Q_{i, (z-1)} \le \rho$, we have that $R_{i, z} = 1$ and thus we get $\hat Q_{i, z} = R_{i, z} \hat Q_{i, z} < \rho^{1-\frac 1 {k+1}}$ (using our contrary assumption). 

{\em Inductive step:} Now, we assume the hypothesis holds for some $z$ and prove it holds for $z+1$. We could choose $s=z+1$, in which case $R_{i, {(z+1)}} = \rho / \hat Q_{i, z }$	. Since $R_{i, {(z+1)}} \hat Q_{i, {(z+1)}} <  \rho^{1-\frac 1 {k+1}}$, we get that 
$$\hat Q_{i, {(z+1)}} < \frac { \rho^{1-\frac 1 {k+1}} } {R_{i, {(z+1)}}} <
\frac { \hat Q_{i, z } \rho^{1-\frac 1 {k+1}} } {\rho} < 
\frac { \rho^{\frac {k - (z - j^*)} {k+1}} \rho^{1-\frac 1 {k+1}} } {\rho} = \rho^{\frac {k - (z+1 - j^*)} {k+1}}$$ which completes the proof of the induction.

We can now reach a contradiction by seeing that $\hat Q_{i, k} = 1 <  \rho^{\frac {k -  (k - j^*) } {k+1}} = 
\rho^{\frac { j^*} {k+1}} \le 1$.
\end{proof}

Applying the lemma with $\rho =\frac 1{(k i)^{1+1/k}}$ gives $R_{i, \cdot}$ such that $D_i \le \frac 1 {(k i)^{1+1/k}}$ and $C_i \le (1-\frac 1 {k i})$ for $i \ge 2$. We have that:
\begin{align*}
\prod_{j=1}^i (1-D_j) 
&\ge \prod_{j=2}^i \exp(- \frac {D_j} {1- D_j}) \ge \exp(- 2 \sum_{j=2}^i  D_j) & \textrm{since } {D_j  < \frac 1 2} \\
& \ge \exp(- \frac 2 k \sum_{j=2}^i  j^{- (1+1/k)} )  \ge  \exp(- \frac 2 k \int_{1}^i x^{- (1+1/k)} dx )  \\
&= \exp(- \frac 2 k  [ - k x^{- 1/k} ]_1^i ) \ge e^{ -2 }
\end{align*}

Moreover, we have that

\begin{align*}
\sum_{i=1}^t \prod_{j=i+1}^{t+1} C_j &\le \sum_{i=1}^t \prod_{j=i+1}^{t+1} (1-\frac 1 {k j}) \le \sum_{i=1}^t \exp(-  \sum_{j=i+1}^{t+1} \frac 1 {k j})\\
&\le \sum_{i=1}^t \exp(-  \frac 1 {k} \int_{i+1}^{t+2} \frac 1 {x} dx) \\
&\le \sum_{i=1}^t \exp(-  \frac {\log(\frac{t+2}{i+1})} {k}) = \sum_{i=1}^t \left( \frac {i+1} {t+2} \right)^{1/k} \\
& \le t \left( \frac { \sum_{i=1}^t (i+1)/t} {t+2} \right)^{1/k} & \textrm{Jensen's inequality for } x^{1/k} \\
&= t \left( \frac {t+3} {2 (t+2)} \right)^{1/k} \le t (2/3)^{1/k}
\end{align*}

So overall we have that $Y_t \ge e^{ -2 } t (1- (2/3)^{1/k}) \ge \frac { t } { 3 e^2 k }$.

The expected revenue of the randomized pricing scheme will be at least $\sum_{t=1}^{N-1} kA_tX_t$.\footnote{We ignore the last block since it might have fewer than $k$ agents.} Since each $Y_t =\sum_{i=1}^tX_t\ge \frac { t } { 3 e^2 k }$, we get that the expected revenue by the randomized pricing scheme is at least $ \sum_{t=1}^{N-1} \frac {k A_t } {3 e^2 k} > \sum_{t=3}^N \frac {Rev[\textrm{Block}_t]} {3 e^2 k}$ since $k A_t > Rev[\textrm{Block}_{t+2}]$.

To bound the revenue of the first two blocks we use the single price lemma that says that a single price $p$ is a $2 k$-factor approximation for $2 k$ distributions. If we choose prices $p_1 = ... = p_n = p$ with probability $2/(3 e^2 + 2)$ or the pricing scheme that is defined above with probability $3 e^2/(3 e^2 + 2)$, we get an expected revenue of at least:
$$\frac 2 {3 e^2 + 2} \left( \frac {Rev[\textrm{Block}_1] + Rev[\textrm{Block}_2]} { 2 k } \right) +
\frac {3 e^2} {3 e^2 + 2} \left( \frac {\sum_{t=3}^N Rev[\textrm{Block}_t] } {3 e^2 k } \right) = \frac {\sum_{t=1}^N Rev[\textrm{Block}_t] } {(3 e^2 + 2) k}$$

Since we are randomizing over pricing schemes there exists a single pricing scheme that achieves the necessary approximation. This completes the proof and shows an $O(k)$ approximation.\end{proof}


\subsection{Extension to $m$-goods and position auctions}

We extend the results of the previous section from digital goods, where we have
an unlimited supply of identical goods, to the $m$-unit setting where we have $m$ copies of a good and to position auctions.

\begin{definition}[Position auction]
In a position auction, there are $m$ items are for sale, each with a scale factor $s_j \in [0,1]$. We assume that $s_1 \ge s_2 \ge ... \ge s_m$. The utility of an agent $i$ with value $v_i$ that receives an item $j$ and pays $p$ is equal to $ s_j v_i - p$.
\end{definition}
In sponsored search auctions auctions, the items are slots on a page of search results and the scale factors correspond to the click through rate of the slot.

\begin{theorem}
In any $m$-good or position auction setting, there exists an anonymous mechanism that achieves an approximation of $O(k)$ for $k$-ambiguous distributions.
\end{theorem}

Since $m$-good auctions are a special case of position auctions with $s_j=1$ for all $j$, it suffices to prove this theorem for position auction settings.
Moreover, we can assume w.l.o.g. that $m=n$ since we can always add additional items with $s_j = 0$.

The first step in proving the theorem is getting an upper bound on the revenue of the optimal mechanism.

\begin{lemma}
For any position auction setting with $k$-ambiguous distributions, the maximum achievable revenue is at most $\sum_{i=1}^{k+1} s_1 Rev[F_i] + 
\sum_{i=1}^{n-k-1} s_i a_i$, where $Rev[F_i]$ is the Myerson revenue of agent $i$'s distribution.
\end{lemma}

\begin{proof}
Consider a setting where we have an additional copy of every item and we run 2 auctions instead of one:
\begin{itemize}
\item Auction A: The $k+1$ first agents are participating.
\item Auction B: The $n-(k+1)$ last agents are participating.
\end{itemize}
We claim that the revenue in this setting is not less than the revenue from the original auction by arguing that the following mechanism gets exactly the same revenue as before. Let $M$ be the optimal mechanism in the original setting. Run $M$ in each auction by sampling bids from the distributions of the missing agents. It is easy to see that the first auction achieves revenue equal to the revenue contribution of the first $k+1$ agents in the original auction while the second one achieves revenue equal to the revenue contribution of the last $n-(k+1)$ agents in the original auction. So it suffices to bound the revenue of each of the two auctions separately.

For the auction A, we can see that an auction A' where there are unlimited copies of item 1 (with scale factor $s_1$) would give us at least as much revenue since scale factors of the items could be artificially reduced to match those in auction $A$. Therefore the revenue of the first auction is upper bounded by $\sum_{i=1}^{k+1} s_1 Rev[F_i]$.

For the auction B, we can see that an auction B' where each agent comes from a point distribution with value $a_i$ instead of the original distribution supported in $[a_{i+k+1},b_{i+k+1}]$ would achieve at least as much revenue. This is because, in auction B', the bids of each agent can be completely ignored and resampled from the previous distribution for every agent and then run the optimal mechanism for auction B. This mechanism is definitely DSIC since it doesn't depend at all at the agent bids and is IR since $a_i > b_{i+k+1}$ which means that an agent in auction $B'$ will always afford to pay the asked price. The revenue in auction B' is exactly 
$\sum_{i=1}^{n-k-1} s_i a_i$ which gives us the upper bound for auction B.
\end{proof}

We now try to construct a mechanism that achieves good approximation guarantees compared to the bound on revenue we've proven. We alter slightly the rules of the decreasing price mechanism with the following rule. Agents that are asked to pay price $p_j$ in the original DPM will get item $j$ and pay $s_j p_j$. 

Notice that although more than 1 agent may be assigned to item $j$, the effect of item $j$ can be simulated by giving an item with higher scale factor $j' < j$ to each additional agent but with probability $s_j / s_{j'}$. Since for every $j$ there are at most $j$ agents assigned to items 1 through $j$ this mechanism is feasible. This mechanism is also DSIC since if agent $i$ is priced $p_j$ any bid higher than $p_j$ gets him exactly the same price and allocation while any bid lower than $p_j$ gets him dropped of the mechanism without any item.
Therefore bidding his real value is preferable since the mechanism is clearly IR.

We use the randomized construction of the previous section to create the DPM mechanism. Under this construction, a price of at least $A_t$ is assigned to $\frac t {3 e^2 k}$ blocks in expectation which gives a revenue at least $\sum_{t=1}^{N-1} k A_t S_t / (3 e^2 k)$ where $S_t = s_{t k}$ is the scale factor of the item agents that are priced $A_t$ receive. 
Since $k A_t S_t \ge \sum_{i=t k }^{t k+k-1} s_i a_i$, we get revenue of at least $\sum_{i=k}^{n-k-1} \frac {s_i a_i} {3 e^2 k}$.

To bound the remaining terms of the revenue, we use second price auction with a single price $p$ to sell just the first item. With probability $1/2$ we set $p$ to be the Myerson reserve of the first agent while with probability $1/(2k)$ we set $p$ to be the the reserve price of the $i$-th agent for $i=2$ to $k+1$. This is an anonymous mechanism and gets
revenue at least $ s_1 (k  Rev[F_1] + \sum_{i=2}^{k+1}  Rev[F_i] ) / (2 k)$.

If we run the DPM mechanism with probability $2/(3 e^2 + 2)$ or the second price auction with probability $3 e^2/(3 e^2 + 2)$, we get an expected revenue of at least:
$$\frac 2 {3 e^2 + 2} \left( \frac { s_1 (k  Rev[F_1] + \sum_{i=2}^{k+1}  Rev[F_i] )} { 2 k } \right) +
\frac {3 e^2} {3 e^2 + 2} \left( \sum_{i=k}^{n-k-1} \frac {s_i a_i} {3 e^2 k} \right) $$
which is at least
$$\frac {\sum_{i=1}^{k+1} s_1 Rev[F_i] + 
\sum_{i=1}^{n-k-1} s_i a_i} {(3 e^2 + 2) k}$$
This completes the proof and shows an $O(k)$ approximation.


\section{Conclusion}

Anonymity imposes real constraints on an auction and, as we have seen, on the revenue it can achieve. In the worst case, we have shown that anonymous mechanisms are quite limited, and that the best anonymous mechanism cannot substantially beat a simple single price. The real advantage of an anonymous mechanism is directly related to the auctioneer's ability to infer information about $f_i$ and $v_i$ from the bids of other advertisers, $\vec v_{-i}$, in essence circumventing the ex-ante anonymity requirement.

Our work leaves a few immediate open questions about anonymous auctions with limited ambiguity. We showed that anonymous auctions can achieve a $\Theta(k)$ approximation for general $k$-ambiguous distributions. For single price mechanisms, we saw that the worst-case approximation improves from $\Theta(n)$ to $\Theta(\log n)$ when distributions are regular --- {\em can we show an analogous $\Theta(\log k)$ bound in the $k$-ambiguous setting when distributions are regular?} Another interesting research direction is to identify alternative metrics for measuring ambiguity. For example, {\em what can we say about the revenue from an anonymous auction when the differential entropy between $f_i$ and the inferred posterior $h$ is small?}

More broadly, our work suggests many general questions about anonymous mechanisms. {\em Can anonymous auctions achieve good approximations beyond the settings we have studied?} Interesting dependencies arise outside the digital goods setting because one bidder's bid can affect the auctioneer's inference about another bidder, affecting the outcome of the auction in a complicated way.  Another question is one of computational complexity --- {\em how difficult is it to compute the optimal anonymous auction?}

%

\bibliographystyle{plain}

\bibliography{references}

\appendix


\section{Characterizing Optimal Anonymous Mechanisms}\label{sec:opt-char-appendix}

In this section, we generalize our characterization theorem from Section~\ref{sec:opt-char} in the special case where the inferred poster $h$ is regular. Recall the observations we made:
\begin{observation}[Observation~\ref{obs:opt-uniform-permutation}]
The optimal anonymous mechanism remains optimal if we randomly rename bidders before running the auction.
\end{observation}
\begin{observation}[Observation~\ref{obs:opt-symmetric-beliefs}]
Suppose that prior beliefs $\vec F$ are symmetric (possibly correlated). Then there exists a symmetric mechanism that maximizes revenue.
\end{observation}
These observations immediately lead to the following claim that reduces symmetric optimization to general optimization:
\begin{claim}[Claim \ref{clm:symmetric-general-reduction}]Any mechanism that is optimal among dominant strategies IC and ex-post IR mechanisms for the symmetric distribution
\[g(\vec x)=\frac {1} {n!} \sum_{\pi \in \Pi(n)} \prod_{i \in N} f_i( x_{\pi_i})\]
can be transformed into a mechanism that is optimal among symmetric, DSIC, and ex-post IR auctions for the beliefs $\vec F$ by relabeling bidders according to a uniformly random permutation.
\end{claim}

To understand the optimal symmetric mechanism, we use the following theorem of Roughgarden and Talgam-Cohen~\cite{RT13}:
\begin{theorem}[Roughgarden and Talgam-Cohen~\cite{RT13}]
In a private, correlated values setting, the optimal mechanism $\mathcal M=(\mathcal A,\mathcal P)$ that is both DSIC and ex-post IR is the following, as long as $\phi$ is monotone ($h$ is regular):
\begin{enumerate}
\item Elicit values $\vec v$ from the bidders.
\item For each bidder $i$, infer a posterior belief $h(v_i|\vec v_{-i})$ about the distribution of $v_i$ from other bidders values $\vec v_{-i}$. Let $\phi^{h|\vec v_{-i}}(v_i)=v_i-\frac{1-H(v|\vec v_{-i})}{h(v_i|v_{-i})}$ denote the virtual value of a bidder with respect to this inferred distribution.
\item Maximize virtual value $\sum_i\phi^{h|\vec v_{-i}}(v_i)$ under the posterior beliefs and charge according to $\mathcal P(\vec v)=v_i\mathcal A_i(\vec v)-\int_0^{v_i}\mathcal A_i(\vec v_{-i},z)dz$ (i.e. traditional single-parameter payments~\cite{M81,AT01}).
\end{enumerate}
\end{theorem}
Given this characterization and the preceding claim, our general characterization theorem is immediate:

\begin{theorem}
The optimal anonymous mechanism is the following, as long as $\phi$ is monotone ($h$ is regular):
\begin{enumerate}
\item Elicit values $\vec v$ from the bidders.
\item Imagine that values were drawn from a correlated distribution $g(\vec v)$ where
\[g(\vec x)=\frac {1} {n!} \sum_{\pi \in \Pi(n)} \prod_{i \in N} f_i( x_{\pi_i})\]
and, for each bidder $i$, use the values $\vec v_{-i}$ of other bidders to infer a posterior belief $h(v_i|\vec v_{-i})=\frac{g(\vec v_{-i},v_i)}{\int_0^\infty g(\vec v_{-i},z)dz}$ about the distribution of $v_i$. Let $\phi^{h|\vec v_{-i}}(v_i)=v_i-\frac{1-H(v|\vec v_{-i})}{h(v_i|v_{-i})}$ denote the virtual value of a bidder with respect to this inferred distribution.
\item Maximize virtual value $\sum_i\phi^{h|\vec v_{-i}}(v_i)$ under the posterior beliefs and charge according to $\mathcal P(\vec v)=v_i\mathcal A_i(\vec v)-\int_0^{v_i}\mathcal A_i(\vec v_{-i},z)dz$ (i.e. traditional single-parameter payments~\cite{M81,AT01}).
\end{enumerate}
\end{theorem}

The characterization we gave in Section~\ref{sec:opt-char} of the optimal symmetric auction for digital goods is then an immediate corollary:
\begin{corollary}
The optimal anonymous digital goods auction sets the optimal price for each bidder according to the posterior belief $h$.
\end{corollary}

Also, the extreme cases noted in Section~\ref{sec:opt-char} behave similarly:
\begin{corollary}
If the distributions $f_i$ are point distributions (bidders' values are known precisely to the auctioneer), have non-overlapping support, or are the same for all bidders, then the optimal anonymous mechanism coincides with Myerson's optimal mechanism.
\end{corollary}
In all three cases, the posterior distribution inferred from $\vec v_{-i}$ is precisely $f_i$, therefore the auction precisely identifies each bidder and runs the optimal auction.

\end{document}